\documentclass[11pt]{article}
\usepackage{amsfonts, mathrsfs}
\usepackage{amsmath, verbatim, amsthm}

\newif\ifabstract
\abstracttrue
\abstractfalse
\newif\iffull
\ifabstract \fullfalse \else \fulltrue \fi

\iffull
\newtheorem{lemma}{Lemma}

\fi

\def\boxit#1{\vbox{\hrule\hbox{\vrule\kern3pt
  \vbox{\kern3pt#1\kern3pt}\kern3pt\vrule}\hrule}} 
\def\Box{\boxit{\null}} 

\iffull
\newcommand{\proofend}{ }
\else
\newcommand{\proofend}{\hfill\Box}
\fi

{\unskip\nobreak\hskip 1em plus 1fil\nobreak$\Box$
\parfillskip=0pt%
\endtrivlist}
{\unskip\nobreak\hskip 2em plus 1fil\nobreak$\Box$
\parfillskip=0pt%
\endtrivlist}
{\unskip\nobreak\hskip 2em plus 1fil\nobreak$\Box$
\parfillskip=0pt%
\endtrivlist}

\newcommand{\junk}[1]{}
\let\bsection=\section
\let\bsubsection=\subsection

\let\lsubsubsection=\subsubsection
\let\bparagraph=\paragraph

\newenvironment{enumroman}
{

\begin{enumerate}}
{\end{enumerate}}

\newenvironment{enumalpha}
{

\begin{enumerate}}
{\end{enumerate}}

\newcommand{\set}[1]{\left\{#1\right\}}
\newcommand{\card}[1]{\left\vert#1\right\vert}

\newcommand{\Oh}[1]{O \left ( #1 \right )}
\newcommand{\flo}[1]{\left\lfloor #1 \right \rfloor}

\newcommand{\Thetah}[1]{\Theta \left ( #1 \right )}
\newcommand{\oh}[1]{o \left ( #1 \right )}

\newcommand{\rangeexp}{c}
\newcommand{\sequence}{S}
\newcommand{\sequencep}{\sequence'}
\newcommand{\sequencet}{\sequence'''}
\newcommand{\sequences}{\sequence''}
\newcommand{\subrangeexp}{\epsilon}
\newcommand{\polyexp}{\beta}
\newcommand{\polydist}{\delta}
\newcommand{\distfan}{t}

\newcommand{\subrange}{r}

\newcommand{\memfactor}{\alpha}
\newcommand{\memsize}{m}
\newcommand{\enca}[1]{\mathbb{#1}}
\newcommand{\seta}[1]{\mathcal{#1}}

\newcommand{\distel}{d}
\newcommand{\unitsize}{u}
\newcommand{\finalunitsize}{f}
\newcommand{\finalunitnum}{s}
\newcommand{\numaggr}{k}
\newcommand{\polylog}[1]{\,\mbox{polylog}\left(#1\right)}

\newcommand{\inlinestep}[1]{\textit{#1}}

\newcommand{\hi}{\mathrm{hi}}
\newcommand{\lo}{\mathrm{lo}}

\newcommand{\twodots}{\ldots}   
\renewcommand{\lg}{\log}

\def\compactify{ }
\let\latexusecounter=\usecounter
\newenvironment{itemize*}
  {\def\usecounter{\compactify\latexusecounter}
   \begin{itemize}}
  {\end{itemize}\let\usecounter=\latexusecounter}
\newenvironment{enumerate*}
  {\def\usecounter{\compactify\latexusecounter}
   \begin{enumerate}}
  {\end{enumerate}\let\usecounter=\latexusecounter\smallskip}

\title{Radix Sorting With No Extra Space}

\iffull

\author{
 Gianni\ Franceschini \\ Univ. of Pisa \\ {\tt francesc@di.unipi.it}
\and 
 S. Muthukrishnan \\ Google Inc., NY \\ {\tt muthu@google.com}
\and
 Mihai P\v{a}tra\c{s}cu \\ MIT \\ {\tt mip@mit.edu}
}

\else

\author{
     Gianni\ Franceschini\inst{1}
\and S. Muthukrishnan\inst{2}
\and Mihai P\v{a}tra\c{s}cu\inst{3}
}

\institute{
     Dept of Computer Science, Univ. of Pisa; \email{francesc@di.unipi.it}
\and Google Inc., NY; \email{muthu@google.com}\\
\and MIT, Boston; \email{mip@mit.edu}
}

\fi

\begin{document}
\maketitle

\begin{abstract}
It is well known that $n$ integers in the range $[1,n^c]$ can be
sorted in $O(n)$ time in the RAM model using radix sorting. More
generally, integers in any range $[1,U]$ can be sorted in
$O(n\sqrt{\log\log n})$ time~\cite{HT}.  However, these algorithms use
$O(n)$ words of extra memory. Is this necessary?

We present a simple, stable, integer sorting algorithm for words of
size $O(\log n)$, which works in $O(n)$ time and uses only $O(1)$
words of extra memory on a RAM model. This is the integer sorting case
most useful in practice.  We extend this result with same bounds to
the case when the keys are read-only, which is of theoretical
interest. Another interesting question is the case of arbitrary $c$.
Here we present a black-box transformation from any RAM sorting
algorithm to a sorting algorithm which uses only $O(1)$ extra space
and has the same running time.  This settles the complexity of
in-place sorting in terms of the complexity of sorting.
\end{abstract}

\bsection{Introduction}\label{sec:intro}

Given $n$ integer {\em keys} $S[1\ldots n]$ each in the range $[1,n]$, they can 
be sorted in $O(n)$ time using $O(n)$ space by {\em bucket sorting}.
This can be extended to the case when the keys are in the range $[1,n^c]$ for some
positive constant $c$ by {\em radix sorting} that uses repeated bucket
sorting with $O(n)$ ranged keys. The crucial point is to do each bucket sorting 
{\em stably}, that is, if positions $i$ and $j$, $i <j$, had the same key $k$,
then the copy of $k$ from position $i$ appears before that from position $j$
in the final sorted order. Radix sorting takes $\Oh{\rangeexp n}$ time and $O(n)$ space.
More generally, RAM sorting with integers in the range $[1,U]$ is a much-studied problem.
Currently, the best known bound is the randomized algorithm in~\cite{HT} that takes 
$O(n\sqrt{\log \log n})$ time, and the deterministic algorithm in~\cite{A} that takes
$O(n\log\log n)$ time. These algorithms also use $O(n)$ words of extra memory 
in addition to the input. 

We ask a  basic question: {\em do we need $O(n)$ auxiliary space for 
integer sorting}? The ultimate goal would be to design {\em in-place}
algorithms for integer sorting that uses only $O(1)$ extra words, 
This question has been explored in depth for comparison-based sorting,
and after a series of papers, we now know that in-place, stable 
comparison-based sorting can be done in $O(n\log n)$ time~\cite{SS87}. 
Some very nice algorithmic techniques have been developed in this quest.
However, no such results are
known for the integer sorting case. Integer sorting is used 
as a subroutine in a number of algorithms that deal with trees and graphs,
including, in particular, sorting the transitions of a finite state machine. 
Indeed, the problem arose in that context for us. In these applications,
it is  useful if one can sort in-place in $O(n)$ time.
From a theoretical perspective, it is likewise interesting to know
if the progress in RAM sorting, including~\cite{FredmanWillard,A,HT},
really needs extra space.

Our results are in-place algorithms for integer sorting. Taken
together, these results solve much of the issues with space efficiency
of integer sorting problems. In particular, our contributions are
threefold.

\bparagraph{A practical algorithm.}
In Section~\ref{sec:simple}, we present a stable integer sorting algorithm for
$O(\log n)$ sized words that takes $O(n)$ time and uses
only $O(1)$ extra words. 

This algorithm is a simple and practical replacement to radix sort. In the numerous applications where 
radix sorting is used, this algorithm can be used to improve
the space usage from $O(n)$ to only $O(1)$ extra words. 
We have implemented the algorithm with positive results.

One key idea of the algorithm is to compress a portion of the input,
modifying the keys. The space thus made free is used as extra space
for sorting the remainder of the input.

\bparagraph{Read-only keys.}
It is theoretically interesting if integer sorting can be performed
in-place {\em without modifying the keys}. The algorithm above
does not satisfy this constraint. In Section~\ref{sec:unstable}, we present a more sophisticated
algorithm that still takes linear time and uses only $O(1)$ extra
words without modifying the keys. 
In contrast to the previous algorithm, we cannot create space for 
ourselves by compressing keys. Instead, we 
introduce a new technique of {\em pseudo pointers} which we believe will find
applications  in other succinct data structure problems. 
The technique is based on keeping a set of distinct keys as a pool
of preset read-only pointers in order to 
maintain linked lists as in bucket sorting. 

As a theoretical exercise, 
\iffull
in Section~\ref{sec:stable}, we also consider
\else
the full version of this paper also considers
\fi
the case when this sorting has to be done stably. We present an 
algorithm with identical performance that is also stable. 
 Similar to the other in-place 
stable sorting algorithms e.g., 
comparison-based sorting~\cite{SS87}, this algorithm is 
quite detailed and needs very careful management of keys as they
are permuted. The resulting algorithm is likely not of practical value,
but it is still fundamentally important  to know that bucket and 
radix sorting can indeed be solved stably in $O(n)$ time with only
$O(1)$ words of extra space. For example, 
even though comparison-based sorting has been well studied at least since
60's, it was not until much later that  optimal, stable in-place
comparison-based sorting was developed~\cite{SS87}.

\bparagraph{Arbitrary word length.}

Another question of fundamental theoretical interest is whether the recently
discovered integer sorting algorithms that work with long keys and sort in
$o(n\log n)$ time, such as~\cite{FredmanWillard,A,HT}, need
any auxiliary space. In Section~\ref{sec:RAM}, we present a  black-box
transformation from any RAM sorting algorithm to an  sorting
algorithm  which uses only $O(1)$ extra space, and retains the same time bounds. 
As a result, the running time bounds of~\cite{A,HT} can now be matched
with only $O(1)$ extra space. This transformation relies on  a fairly
natural technique of compressing a
portion of the input to make space for simulating space-inefficient RAM sorting
algorithms.

\bparagraph{Definitions.}

Formally, we are given a sequence $\sequence$ of $n$ elements. 
The problem is to sort $\sequence$ according to the integer keys,
under the following assumptions:

\newcounter{opseq}%
\setcounter{opseq}{0}%
\renewcommand{\theopseq}{\alph{opseq}}%
\begin{enumroman}
\item Each element has an integer key
within the interval $[1,U]$.
\item\label{prob:main:opseq} The following unit-cost operations are allowed on
$\sequence$:
\refstepcounter{opseq} (\theopseq) indirect address of any position of
$\sequence$; 
\refstepcounter{opseq} (\theopseq)\label{prob:main:opseq:readonly} read-only
access to the key of any element; 
\refstepcounter{opseq} (\theopseq) exchange of the positions of any two
elements. 
\item The following unit-cost operations are allowed on integer values of
$O(\log U)$
bits: addition, subtraction, bitwise AND/OR and unrestricted bit shift. 
\item\label{prob:main:cond:memory} Only $O(1)$ auxiliary words of memory are
allowed; each word had $\log U$ bits.
\end{enumroman}


For the sake of presentation, we will refer to the elements' keys as
if they were the input elements. For example, for any two elements $x$,
$y$, instead of writing that the key of $x$ is less
than the key of $y$ we will simply write $x<y$.
We also need a precise definition of the {\em rank} of
an element in a sequence when multiple occurrences of keys are allowed: 
the \emph{rank} of an element $x_i$ in a sequence $x_1\ldots x_t$ 
is the cardinality of the multiset
$\left\{x_j\;\vert\; x_j<x_i \mbox{ or } (x_j=x_i \mbox{ and }
  j\le i)\right\}$.

\bsection{Stable Sorting for Modifiable Keys}
\label{sec:simple}

We now describe our simple algorithm for (stable) radix
sort without additional memory. 

\bparagraph{Gaining space.}
The first observation is that numbers
in sorted order have less entropy than in arbitrary order. In
particular, $n$ numbers from a universe of $u$ have binary entropy
$n\lg u$ when the order is unspecified, but only $\lg \binom{u}{n} =
n\lg u - \Theta(n \lg n)$ in sorted order. This suggests that we can
``compress'' sorted numbers to gain more space:

\begin{lemma}  \label{lem:compress}
A list of $n$ integers in sorted order can be represented as: (a) an
array $A[1\twodots n]$ with the integers in order; (b) an array of $n$
integers, such that the last $\Theta(n\lg n)$ bits of the array are
zero. Furthermore, there exist in-place $O(n)$ time algorithms for
switching between representations (a) and (b).
\end{lemma}

\begin{proof}
One can imagine many representations (b) for which the lemma is
true. We note nonetheless that some care is needed, as some obvious
representations will in fact not lead to in-place encoding. Take for
instance the appealing approach of replacing $A[i]$ by $A[i] -
S[i-1]$, which makes numbers tend to be small (the average value is
$\frac{u}{n}$). Then, one can try to encode the difference using a
code optimized for smaller integers, for example one that represents a
value $x$ using $\lg x + O(\lg\lg x)$ bits. However, the obvious
encoding algorithm will not be in-place: even though the scheme is
guaranteed to save space over the entire array, it is possible for
many large values to cluster at the beginning, leading to a rather
large prefix being in fact expanded. This makes it hard to construct
the encoding in the same space as the original numbers, since we need
to shift a lot of data to the right before we start seeing a space
saving.

As it will turn out, the practical performance of our radix sort is
rather insensitive to the exact space saving achieved here. Thus, we
aim for a representation which makes in-place encoding particularly
easy to implement, sacrificing constant factors in the space saving.

First consider the most significant bit of all integers. Observe that
if we only remember the minimum $i$ such that $A[i] \ge u/2$, we know
all most significant bits (they are zero up to $i$ and one after
that). We will encode the last $n/3$ values in the array more
compactly, and use the most significant bits of $A[1\twodots
\frac{2}{3}n]$ to store a stream of $\frac{2}{3} n$ bits needed by the
encoding.

We now break a number $x$ into $\hi(x)$, containing the upper $\lfloor
\log_2 (n/3) \rfloor$ bits, and $\lo(x)$, with the low $\lg u -
\lfloor \log_2 (n/3) \rfloor$ bits. For all values in $A[\frac{2}{3}n
+ 1 \twodots n]$, we can throw away $\hi(A[i])$ as follows. First we
add $\hi(A[\frac{2}{3} n + 1])$ zeros to the bit stream, followed by a
one; then for every $i= \frac{2}{3}n + 2, \dots, n$ we add $\hi(A[i])
- \hi(A[i-1])$ zeros, followed by a one.  In total, the stream
contains exactly $n/3$ ones (one per element), and exactly $\hi(A[n])
\le n/3$ zeros. Now we simply compact $\lo(A[\frac{2}{3} n + 1]),
\dots, \lo(A[n])$ in one pass, gaining $\frac{n}{3} \lfloor
\log_2 (n/3) \rfloor$ free bits.
\proofend
\end{proof}

\bparagraph{An unstable algorithm.}
Even just this compression observation is enough to give a simple algorithm,
whose only disadvantage is that it is unstable. 
The algorithm has the following structure:

\begin{enumerate*}
\item sort the subsequence $S[1\ldots(n/\log n)]$ using the
  optimal in-place mergesort in \cite{SS87}. 

\item compress $S[1\twodots (n/\lg n)]$ by Lemma~\ref{lem:compress},
  generating $\Omega(n)$ bits of free space.

\item radix sort $S[(n/\log n)+1\ldots n]$ using the free space.

\item uncompress $S[1\twodots (n/\lg n)]$.

\item merge the two sorted sequences $S[1\ldots(n/\log n)]$ and
  $S[(n/\log n)+1\ldots n]$ by using the in-place, linear time merge
  in~\cite{SS87}.
\end{enumerate*}

\noindent
The only problematic step is 3. The implementation of this step is
 based on the cycle leader permuting approach where a
sequence $A$ is re-arranged by following the cycles of a permutation $\pi$.
First $A[1]$ is sent in its
final position $\pi(1)$. Then, the element that was in $\pi(1)$ is
sent to its final position $\pi(\pi(1))$. The process proceeds in this way until
the cycle is closed, that is until the element that is moved in position $1$ is
found. At this point, the elements starting from $A[2]$ are scanned until
a new cycle leader $A[i]$ (i.e. its cycle has not been
walked through) is found, $A[i]$'s cycle is followed in its turn, and so forth. 

To sort, we use $2n^\subrangeexp$
counters $c_1,\ldots,c_{n^\subrangeexp}$ and $d_1,\ldots,d_{n^\subrangeexp}$.
They are stored in the auxiliary words obtained in step 2.
Each $d_j$ is initialized to $0$. With a first scan of the elements, we
store in any $c_i$ the number of occurrences of key $i$. Then, for each
$i=2\ldots n^\subrangeexp$, we set $c_i=c_{i-1}+1$ and finally we set $c_1=1$
(in the end, for any $i$ we have that $c_i=\sum_{j < i} c_j + 1$).
Now we have all the information for the cycle leader process. Letting
$j=(n/\log n)+1$, we proceed as follows:
\begin{enumerate*}
\item[$(i)$] let $i$ be the key of $S[j]$;
\item[$(ii)$] if $c_i\le j<c_{i+1}$ then $S[j]$ is already in its final
  position, hence we increment $j$ by $1$ and go to step $(i)$;
\item[$(iii)$] otherwise, we exchange
  $S[j]$ with $S[c_i+d_i]$, we increment $d_i$ by $1$ and we go to step $(i)$. 
\end{enumerate*}

\noindent
Note that this algorithm is inherently unstable, because we cannot
differentiate elements which should fall between $c_i$ and $c_{i+1} -
1$, given the free space we have.

\bparagraph{Stability through recursion.}
To achieve stability, we need more than $n$ free bits, which we can achieve
by bootstrapping with our own sorting algorithm, instead of merge sort.
There is also an important practical advantage to the new stable approach:
the elements are permuted much more conservatively, resulting in better cache performance.
\begin{enumerate*}
\item recursively sort a constant fraction of the array, say
  $S[1\twodots n/2]$.

\item compress $S[1\twodots n/2]$ by Lemma~\ref{lem:compress},
  generating $\Omega(n\lg n)$ bits of free space.

\item for a small enough constant $\gamma$, break the remaining $n/2$
  elements into chunks of $\gamma n$ numbers. Each chunk is sorted by
  a classic radix sort algorithm which uses the available space.

\item uncompress $S[1\twodots n/2]$.

\item we now have $1 + 1/\gamma = O(1)$ sorted subarrays. We merge
  them in linear time using the stable in-place algorithm of
  \cite{SS87}.
\end{enumerate*}

\noindent
We note that the recursion can in fact be implemented bottom up, so
there is no need for a stack of superconstant space. For the base
case, we can use bubble sort when we are down to $n \le \sqrt{n_0}$
elements, where $n_0$ is the original size of the array at the top
level of the recursion.

Steps 2 and 4 are known to take $O(n)$ time. For step 3, note that
radix sort in base $R$ applied to $N$ numbers requires $N+R$
additional words of space, and takes time $O(N \log_R u)$. Since we
have a free space of $\Omega(n\lg n)$ bits or $\Omega(n)$ words, we
can set $N = R = \gamma n$, for a small enough constant $\gamma$. As
we always have $n = \Omega(\sqrt{n_0}) = u^{\Omega(1)}$, radix sort
will take linear time.

The running time is described by the recursion $T(n) = T(n/2) +
O(n)$, yielding $T(n) = O(n)$.

\bparagraph{A self-contained algorithm.}
Unfortunately, all algorithms so far use in-place stable merging
algorithm as in~\cite{SS87}.
We want to remove this dependence, and obtain a simple and practical sorting algorithm.
By creating free space through compression
at the right times, we can instead use a simple merging implementation that needs
additional space. We first observe the following:

\begin{lemma} \label{lem:merge}
Let $k \ge 2$ and $\alpha > 0$ be arbitrary constants.  Given $k$
sorted lists of $n/k$ elements, and $\alpha n$ words of free space, we
can merge the lists in $O(n)$ time.
\end{lemma}

\begin{proof}
We divide space into blocks of $\alpha n / (k+1)$ words. Initially, we
have $k+1$ free blocks. We start merging the lists, writing the output
in these blocks. Whenever we are out of free blocks, we look for
additional blocks which have become free in the original sorted lists.
In each list, the merging pointer may be inside some block, making it
yet unavailable. However, we can only have $k$ such partially consumed
blocks, accounting for less than $k \frac{\alpha n}{k+1}$ wasted words
of space. Since in total there are $\alpha n$ free words, there must
always be at least one block which is available, and we can direct
further output into it.

At the end, we have the merging of the lists, but the output appears
in a nontrivial order of the blocks. Since there are $(k+1) (1 +
1/\alpha) = O(1)$ blocks in total, we can remember this order using
constant additional space. Then, we can permute the blocks in linear
time, obtaining the true sorted order.
\proofend
\end{proof}

Since we need additional space for merging, we can never work with the
entire array at the same time. However, we can now use a classic
sorting idea, which is often used in introductory algorithms courses
to illustrate recursion (see, e.g.~\cite{CLR}). To sort $n$ numbers,
one can first sort the first $\frac{2}{3} n$ numbers (recursively),
then the last $\frac{2}{3} n$ numbers, and then the first $\frac{2}{3}
n$ numbers again. Though normally this algorithm gives a running time
of $\omega(n^2)$, it works efficiently in our case because we do not
need recursion:

\begin{enumerate*}
\item sort $S[1\twodots n/3]$ recursively.

\item compress $S[1\twodots \frac{n}{3}]$, and sort $S[\frac{n}{3}
  +1 \twodots n]$ as before: first radix sort chunks of $\gamma n$
  numbers, and then merge all chunks by Lemma~\ref{lem:merge} using
  the available space. Finally, uncompress $S[1\twodots \frac{n}{3}]$.

\item compress $S[\frac{2n}{3} + 1\twodots n]$, which is now
  sorted. Using Lemma~\ref{lem:merge}, merge $S[1\twodots
  \frac{n}{3}]$ with $S[\frac{n}{3}+1 \twodots \frac{2n}{3}]$. Finally
  uncompress. 

\item once again, compress $S[1\twodots \frac{n}{3}]$, merge
  $S[\frac{n}{3} + 1 \twodots \frac{2n}{3}]$ with $S[\frac{2n}{3}+1
  \twodots n]$, and uncompress.
\end{enumerate*}

\noindent
Note that steps 2--4 are linear time. Then, we have the recursion
$T(n) = T(n/3) + O(n)$, solving to $T(n) = O(n)$. Finally, we note
that stability of the algorithm follows immediately from stability of
classic radix sort and stability of merging.

\bparagraph{Practical experience.}
The algorithm is surprisingly effective in practice. It can be
implemented in about 150 lines of C code. Experiments with sorting
1-10 million 32-bit numbers on a Pentium machine indicate the
algorithm is roughly 2.5 times slower than radix sort with additional
memory, and slightly faster than quicksort (which is not even stable).

\section{Unstable Sorting for Read-only Keys}
\label{sec:unstable}


\bsubsection{Simulating auxiliary bits}\label{subsubsec:bitstealing}
With the bit stealing technique \cite{Mu86}, a
bit of information is encoded in the relative order of a pair of 
elements with different keys: the pair is maintained in
increasing order to encode a $0$ and vice versa. 
The obvious drawback of this technique is that the cost of accessing a
word of $w$ encoded bits is $\Oh{w}$ in the worst case (no word-level
parallelism). However, if we
modify an encoded word with a series of $l$ increments (or decrements) by
$1$, the total cost of the entire series is $\Oh{l}$ (see~\cite{CLR}). 

To find pairs of distinct elements,
we go from $\sequence$ to a sequence
$Z'Y'XY''Z''$ with two properties. $(i)$ For any $z'\in Z'$, 
$y'\in Y'$, $x\in X$, $y''\in Y''$  and 
$z''\in Z''$ we have that $z'<y'<x<y''<z''$. 
$(ii)$ Let $\memsize=\memfactor\left\lceil n/\log n\right\rceil$, for a
suitable constant $\memfactor$. $Y'$ is composed by the element
$y_\memsize'$ with rank $\memsize$ plus all the other elements equal to
$y_\memsize'$. $Y''$ is composed by the element $y_\memsize''$ with rank
$n-\memsize+1$ plus all the other elements equal to $y_\memsize''$.
To obtain the new sequence we use the in-place, linear time selection and
partitioning algorithms in \cite{KP92,KP94}. 
If $X$ is empty, the task left is to sort
$Z'$ and $Z''$, which can be accomplished with any optimal, in-place
mergesort (e.g. \cite{SS87}.
Let us denote $Z'Y'$ with $M'$ and 
$Y''Z''$ with $M''$. The $\memsize$ pairs of
distinct elements
$(M'[1],M''[1]),(M'[2],M''[2]),\ldots,(M'[\memsize],M''[\memsize])$ will be used
to encode information.

Since the choice of the constant $\memfactor$ does not affect the asymptotic
complexity of the algorithm, we have reduced our problem 
to a problem in which we are allowed to
use a special \emph{bit memory} with $\Oh{n/\log n}$ bits
where each bit can be accessed and modified in constant time but without 
word-level parallelism.

\bsubsection{Simulating auxiliary memory for permuting}
\label{subsubsec:intbuf}
With the internal buffering
technique \cite{kronrod}, some of the elements are
used as placeholders in order to simulate a working area and permute the other
elements at lower cost. In our unstable sorting algorithm we use the
basic idea of internal buffering in the following way.
Using the selection and partitioning algorithms in
\cite{KP92,KP94}, we pass from the original sequence $S$ to $ABC$ with 
two properties. 
$(i)$ For any $a\in A$, 
$b\in B$ and $c\in C$, we have that $a<b<c$. 
$(ii)$ $B$ is
composed of the element
$b'$ with rank $\lceil n/2\rceil$ plus all the other elements equal to
$b'$. We can use $BC$ as an auxiliary memory in
the following way. The element in the first position of $BC$ is the
\emph{separator element} and will not be moved. The elements in the other
positions of $BC$ are placeholders and will be exchanged with (instead of being
overwritten by) elements from $A$ in any way the computation on $A$ (in our case
the sorting of $A$) may require. The ``emptiness'' of any location $i$
of the simulated working area in $BC$ can be tested in $O(1)$ time by comparing
the separator element $BC[1]$ with $BC[i]$: if $BC[1]\le BC[i]$ the $i$th
location is ``empty'' (that is, it contains a placeholder),
otherwise it contains one of the elements in $A$. 

Let us suppose we can sort the elements in $A$ in $O(\card{A})$ time using $BC$
as working area. After $A$ is sorted we use the partitioning algorithm in
\cite{KP92} to separate the elements equal to the separator element ($BC[1]$)
from the elements greater than it (the computation on $A$ may have altered the
original order in $BC$). Then we just re-apply the same process to $C$, that
is we divide it into $A'B'C'$, we sort $A'$ using $B'C'$ as working area and so
forth. Clearly, this process requires $O(n)$ time and when it
terminates the elements are sorted.  
Obviously, we can divide $A$ into 
$p=\Oh{1}$ equally sized subsequences $A_1,A_2\ldots A_p$, then sort
each one of them using $BC$ as working area and finally fuse
them using the in-place, linear time merging algorithm in \cite{SS87}.
Since the choice of the constant $p$ does not affect the asymptotic
complexity of the whole process, we have reduced
our problem 
 to a new problem, in which we
are allowed to use a special  \emph{exchange} memory of $\Oh{n}$ locations,
where each location can contain \emph{input elements only} (no integers or any
other kind of data). Any element can be moved to and from any location of
the exchange memory in $O(1)$ time.

\bsubsection{The reduced problem}\label{subsubsec:reduced}
By blending together the basic techniques seen above, we can focus on a
reduced problem 
in which assumption~\ref{prob:main:cond:memory} is replaced by:
\textit{%
\begin{enumroman}
\item[\ref{prob:main:cond:memory}] Only $O(1)$ words of
normal auxiliary memory and two kinds of special auxiliary
memory are allowed: 
\begin{enumalpha}
\item A random access \emph{bit memory} $\mathscr{B}$ with
$\Oh{n/\log n}$ bits, where each bit can be accessed in $O(1)$ time
(no word-level parallelism).
\item A random access \emph{exchange memory} $\mathscr{E}$ with $\Oh{n}$
locations, where each location can contain \emph{only elements from $S$} and
they can be moved to and from any location of $\mathscr{E}$
in $O(1)$ time.
\end{enumalpha}
\end{enumroman}
}

\noindent{}If we can solve the reduced problem in $\Oh{n}$ time we
can also solve the original problem 
with the same asymptotic complexity.
However, the resulting algorithm will be 
unstable because of the use of the internal buffering technique with a large
pool of placeholder elements.

\bsubsection{The naive approach}\label{subsubsec:naive}
Despite the two special
auxiliary memories, solving the reduced problem is not easy. 
Let us consider the following naive approach.
We proceed as in the normal bucket sorting: one bucket for each one of the
$n^\subrangeexp$ range values. Each bucket is a linked list: the input elements
of each bucket are maintained in $\mathscr{E}$ while its auxiliary data (e.g.
the pointers of the list) are maintained in $\mathscr{B}$. In order to amortize
the cost of updating the auxiliary data (each pointer requires a word of 
$\Thetah{\log n}$ bits and $\mathscr{B}$ does not have word-level
parallelism), each bucket is a linked list of \emph{slabs} of
$\Thetah{\log^2 n}$ elements each
($\mathscr{B}$ has only $\Oh{n/\log n}$ bits). At any time each bucket has a
partially full \emph{head slab} which is where any new element of
the bucket is stored. Hence, for each bucket we need to store in $\mathscr{B}$ a
word of $\Oh{\log\log n}$ bits with the position in the head slab of the
last element added. The algorithm proceeds as usual: each element in $S$ is 
sent to its bucket in $\Oh{1}$ time and is inserted in the bucket's head slab.
With no word-level parallelism in $\mathscr{B}$ the insertion in the head
slab requires $\Oh{\log\log n}$ time. Therefore, we have an 
$O(n\log\log n)$ time solution for the reduced problem and, consequently, an
unstable $O(n\log\log n)$ time solution for the original problem.

This simple strategy can be improved by dividing the head slab of
a bucket into \emph{second level slabs} of $\Thetah{\log\log n}$ elements
each. As for the first level slabs, there is a partially full, second
level head slab.
For any bucket we maintain two words in $\mathscr{B}$: the first one
has $\Oh{\log\log\log n}$ bits and stores the position of the last element
inserted in the second level head slab; the second one has 
$\Oh{\log\log n}$ bits and stores the position of the last full slab of second
level contained in the first level head slab. Clearly, this gives
us an $O(n\log\log\log n)$ time solution for the reduced problem and the
corresponding unstable solution
for the original problem. 
By generalizing this
approach to the extreme, we end up with $\Oh{\log^* n}$ levels of slabs, an 
$\Oh{n\log^* n}$ time solution for the reduced problem and the related unstable
solution for the original problem. 

\bsubsection{The pseudo pointers}\label{subsubsec:pseudo}
Unlike bit stealing and internal buffering which were known earlier, 
the pseudo pointers technique has been 
specifically designed for improving the space complexity in integer sorting problems.
Basically, in this technique a set of elements
with distinct keys is used as a pool of pre-set, read-only pointers in order to
simulate efficiently traversable and updatable linked lists.  
Let us show how to use this basic idea in a particular procedure that will be at
the core of our optimal solution for the reduced problem.

Let $\distel$ be the number of distinct keys in $S$. We are given
two sets of $\distel$ input elements with distinct keys: the sets $\mathcal{G}$
and $\mathcal{P}$ of \emph{guides} and \emph{pseudo pointers},
respectively. The guides are given us \emph{in sorted
order} while the pseudo pointers form a sequence in arbitrary
order. Finally, we are given a multiset $\mathcal{I}$ of $\distel$ input
elements (i.e. two elements of $\mathcal{I}$ can have equal keys). The procedure
uses the guides, the pseudo pointers and the exchange memory to
sort the $\distel$ input elements of $\mathcal{I}$ in $\Oh{\distel}$ time. 

We use three groups of contiguous locations in the exchange memory
$\mathscr{E}$. The first group $H$ has $n^\subrangeexp$ locations (one for each
possible value of the keys). The second group $L$ has $n^\subrangeexp$
\emph{slots} of two adjacent locations each. The last group $R$ has $\distel$
locations, the elements of $\mathcal{I}$ will end up here in sorted order.
$H$, $L$ and $R$ are initially empty. We have two main steps.

\inlinestep{First}. 
For each $s\in\mathcal{I}$, we proceed as follows. Let $p$ be the leftmost pseudo pointer
still in $\mathcal{P}$. If the
$s$th location of $H$ is empty, we move $p$ from $\mathcal{P}$ to
$H[s]$ and then we move $s$ from $\mathcal{I}$ to the first location of $L[p]$
(i.e. the first location of the $p$th slot of $L$) leaving the second location
of $L[p]$ empty.
Otherwise, if $H[s]$ contains an element $p'$ (a pseudo pointer)
we move $s$ from
$\mathcal{I}$ to the first location of $L[p]$, then we move $p'$ from $H[s]$ to
the second location of $L[p]$ and finally we move $p$ from $\mathcal{P}$ to
$H[s]$. 

\inlinestep{Second}. We scan the guides in $\mathcal{G}$ from the
smallest to the largest one. 
For a guide $g\in\mathcal{G}$ we proceed as follows. If the $g$th
location of $H$ is empty then there does not exist any element equal to $g$
among 
the ones to be sorted (and initially in $\mathcal{I}$) and hence we move to
the next guide. Otherwise, if $H[G]$ contains a pseudo pointer $p$,
there is at least one element equal to $g$ among the ones to be sorted and this
element is currently stored in the first location of the $p$th slot of $L$.
Hence, we move that element from the first location of $L[p]$ to the leftmost
empty location of $R$. After that, if the second location of $L[p]$ contains a
pseudo pointer $p'$, there is another element equal to $g$ and we
proceed in the same fashion. Otherwise, if the second location of $L[p]$ is
empty then there are no more elements equal to $g$ among the ones to be sorted
and therefore we can focus on the next guide element. 

Basically, the procedure is bucket sorting where the auxiliary data of the
list associated to each bucket (i.e. the links among elements in the list)
\emph{is implemented by pseudo pointers in $\mathcal{P}$} instead of
storing it explicitly in the bit memory (which lacks of word-level parallelism
and is inefficient in access). It is worth noting that the buckets' lists 
implemented with pseudo pointers are spread over an area that is larger
than the one we would obtain with explicit pointers (that is because each
pseudo pointer has a key of $\log n^\subrangeexp$ bits while an
explicit pointer would have only $\log \distel$ bits).  

\bsubsection{The optimal solution}\label{subsubsec:algo}
We can now describe the algorithm, which has three main steps.

\inlinestep{First}. Let us assume that for any element $s\in S$ there is at
least another element with the same key. (Otherwise, we can easily reduce
to this case in linear time: we isolate the $\Oh{n^\subrangeexp}$ elements that
do not respect the property, we sort them with the in-place
mergesort in \cite{SS87} and finally we merge them after the other $O(n)$
elements are sorted.) With this assumption, we extract from $S$ two sets  
$\mathcal{G}$ and $\mathcal{P}$ of $\distel$ input elements with distinct keys 
(this can be easily achieved in $\Oh{n}$ time using only the exchange memory
$\mathscr{E}$). Finally we sort $\mathcal{G}$ with the optimal in-place
mergesort in \cite{SS87}. 
  
\inlinestep{Second}. Let $S'$ be the sequence with the ($\Oh{n}$) input elements
left after the first step. Using the procedure in \S~\ref{subsubsec:pseudo}
(clearly, the elements in the sets $\mathcal{G}$ and $\mathcal{P}$ computed in
the first step will be the guides and pseudo pointers used in the procedure), we
sort each block $B_i$ of $S'$ with $\distel$ contiguous elements. After that,
let us focus on the first $\distfan=\Thetah{\log\log n}$ consecutive blocks
$B_1,B_2,\ldots,B_t$. We distribute the elements of these blocks into 
$\le \distfan$ \emph{groups} $G_1,G_2\ldots$ in the following way. Each group
$G_j$ can contain between $\distel$ and $2\distel$ elements and is allocated in
the exchange memory $\mathscr{E}$. The largest element in a group is its
\emph{pivot}. The number of elements in a group is stored in a word of
$\Thetah{\log\distel}$ bits
allocated in the bit memory $\mathscr{B}$. Initially there is only one group and
is empty. In the $i$th step of the distribution we scan the 
elements of the $i$th block $B_i$. As long as the elements of $B_i$ are less
than or equal to the pivot of the first group we move them into it. If, during
the process, the group becomes full, we select its median element and partition
the group into two new groups (using the selection and
partitioning algorithms in \cite{KP92,KP94}). When, during the scan, the
elements of $B_i$ become greater than the pivot of the first group, we move to
the second group and continue in the same fashion. It is important to notice
that the number of elements in a group (stored in a word of
$\Thetah{\log\distel}$ bits in the bit memory $\mathscr{B}$) is updated by
increments by $1$ (and hence the total cost of updating the number of elements
in any group is linear in the final number of elements in that
group, see~\cite{CLR}). Finally, when all the elements of the first 
$\distfan=\Thetah{\log\log n}$ consecutive blocks
$B_1,B_2,\ldots,B_t$ have been distributed into groups, we sort each group using
the procedure in \S~\ref{subsubsec:pseudo} (when a group has more than
$\distel$ elements, we sort them in two batches and then merge them with the
in-place, linear time merging in \cite{SS87}). The whole process is repeated for
the second $\distfan=\Thetah{\log\log n}$ consecutive blocks, and so forth.

\inlinestep{Third}. After the second step, the sequence $S'$ (which contains all
the elements of $S$ with the exclusion of the guides and pseudo pointers, see
the first step) is composed by contiguous subsequences $S'_1,S'_2,\ldots$ which
are \emph{sorted} and contain 
$\Theta(\distel\log\log n)$ elements each (where $\distel$ is the number of
distinct elements in $S$). Hence, if we see $S'$ as composed by
contiguous runs of elements with the same key, we can conclude that the number
of runs of $S'$ is $\Oh{n/\log\log n}$. Therefore $S'$ can be sorted in
$\Oh{n}$ time using the naive approach described in
\S~\ref{subsubsec:naive} with only the following simple modification. 
As long as we are inserting the elements of a single run in a bucket,
we maintain the position of the last element inserted in the head slab of the
bucket in a word of auxiliary memory (we can use $O(1)$ of them) instead of
accessing the inefficient bit memory $\mathscr{B}$ at any single insertion. When
the current run is finally exhausted, we copy the position in the bit memory.
Finally, we sort $\mathcal{P}$  and we merge $\mathcal{P}$,
$\mathcal{A}$ and $S'$ (once again, using the sorting and merging algorithms in
\cite{SS87}).

\bsubsection{Discussion: Stability and Read-only Keys}\label{subsec:reasons}

Let us focus on the reasons why the algorithm of this section 
is not stable. 
The major cause of
instability is the use of the basic 
internal buffering technique in
conjunction with large ($\omega(\polylog{n})$) pools of placeholder
elements. This is clearly visible even in the first
iteration of the process in \S~\ref{subsubsec:intbuf}: 
after being used to permute $A$ into
sorted order, the placeholder elements in $BC$ are left permuted in a
completely arbitrary way and their initial order is lost.

\section{Reducing Space in any RAM Sorting Algorithm}   \label{sec:RAM}

In this section, we consider the case of sorting integers of $w =
\omega(\lg n)$ bits. We show a black box transformation from any
sorting algorithm on the RAM to a stable sorting algorithm with the
same time bounds which only uses $O(1)$ words of additional space. Our
reduction needs to modify keys. Furthermore, it requires randomization
for large values of $w$.

We first remark that an algorithm that runs in time $t(n)$ can only
use $O(t(n))$ words of space in most realistic models of
computation. In models where the algorithm is allowed to write $t(n)$
arbitrary words in a larger memory space, the space can also be
reduced to $O(t(n))$ by introducing randomization, and storing the
memory cells in a hash table.

\paragraph{Small word size.}
We first deal with the case $w = \polylog{n}$. The algorithm has the
following structure:

\begin{enumerate*}
\item sort $S[1\twodots n/\lg n]$ using in-place stable merge sort
  \cite{SS87}. Compress these elements by Lemma \ref{lem:compress}
  gaining $\Omega(n)$ bits of space.

\item since $t(n) = O(n\lg n)$, the RAM sorting algorithm uses at most
  $O(t(n) \cdot w) = O(n\polylog{n})$ bits of space. Then we can break
  the array into chunks of $n / \lg^c n$ elements, and sort each one
  using the available space.

\item merge the $\lg^c n$ sorted subarrays.

\item uncompress $S[1\twodots n/\lg n]$ and merge with the rest of
  the array by stable in-place merging~\cite{SS87}.
\end{enumerate*}

\noindent
Steps 1 and 4 take linear time. Step 2 requires $\lg^c n \cdot t(n /
\lg^c n) = O(t(n))$ because $t(n)$ is convex and bounded in $[n, n\lg
  n]$. We note that step 2 can always be made stable, since we can
afford a label of $O(\lg n)$ bits per value.

It remains to show that step 3 can be implemented in $O(n)$ time. In
fact, this is a combination of the merging technique from
Lemma~\ref{lem:merge} with an atomic heap \cite{fredman94atomic}.
The atomic heap can maintain a priority queue over $\polylog{n}$
elements with constant time per insert and extract-min. Thus, we can
merge $\lg^c n$ lists with constant time per element. The atomic heap
can be made stable by adding a label of $c\lg\lg n$ bits for each
element in the heap, which we have space for. The merging of
Lemma~\ref{lem:merge} requires that we keep track of $O(k/\alpha)$
subarrays, where $k=\lg^c n$ was the number of lists and $\alpha =
1/\polylog{n}$ is fraction of additional space we have
available. Fortunately, this is only $\polylog{n}$ values to record,
which we can afford.

\paragraph{Large word size.}
For word size $w \ge \lg^{1+\varepsilon} n$, the randomized algorithm
of \cite{A} can sort in $O(n)$ time. Since this is
the best bound one can hope for, it suffices to make this particular
algorithm in-place, rather than give a black-box transformation.
We use the same algorithm from above. The only challenge is to make
step 2 work: sort $n$ keys with $O(n\polylog{n})$ space, even if the
keys have $w > \polylog{n}$ bits.

We may assume $w \ge \lg^3 n$, which simplifies the algorithm
of~\cite{A} to two stages. In the first stage, a
signature of $O(\lg^2 n)$ bits is generated for each input value
(through hashing), and these signatures are sorted in linear
time. Since we are working with $O(\lg^2 n)$-bit keys regardless of
the original $w$, this part needs $O(n\polylog{n})$ bits of space, and
it can be handled as above.

From the sorted signatures, an additional pass extracts a subkey of
$w/\lg n$ bits from each input value. Then, these subkeys are sorted
in linear time. Finally, the order of the original keys is determined
from the sorted subkeys and the sorted signatures. 

To reduce the space in this stage, we first note that the algorithm
for extracting subkeys does not require additional space. We can
then isolate the subkey from the rest of the key, using shifts, and
group subkeys and the remainder of each key in separate arrays, taking
linear time. This way, by extracting the subkeys instead of copying
them we require no extra space. We now note that the algorithm
in~\cite{A} for sorting the subkeys also does not
require additional space. At the end, we recompose the keys by
applying the inverse permutation to the subkeys, and shifting them
back into the keys.

Finally, sorting the original keys only requires knowledge of the
signatures and \emph{order} information about the subkeys. Thus, it
requires $O(n \polylog{n})$ bits of space, which we have. At the end, we
find the sorted order of the original keys and we can implement the
permutation in linear time.

\iffull

\section{Stable sorting for read-only keys}\label{sec:stable}

In the following we will denote $n^\subrangeexp$
with $\subrange$. Before we begin, let us recall that two consecutive sequences
$X$ and $Y$,
possibly of different sizes, can be exchanged stably, in-place and in linear
time with three sequence reversals, since $YX=(X^RY^R)^R$. 
Let us give a short overview of the algorithm. We have three phases.

\bparagraph{Preliminary phase (\S~\ref{subsec:prelim}).} 
The purpose of this phase is to obtain some collections of elements
to be used with the three techniques described in
\S~\ref{sec:unstable}. We extract 
 $\Thetah{n/\log n}$ smallest and largest elements of
$S$. They will form an encoded memory of $\Thetah{n/\log n}$ bits.
Then, we extract from the remaining sequence $\Thetah{n^\subrangeexp}$ smallest
elements and divide them into $O(1)$ \emph{jump zones} of equal
length.
After that, we extract from the remaining sequence some 
equally sized sets of distinct elements. Each set is collected into a contiguous
zone. At the end of the phase, we have \emph{guide}, \emph{distribution},
\emph{pointer}, \emph{head} and \emph{spare} zones. 

\bparagraph{Aggregating phase (\S~\ref{subsec:aggr}).}
After the preliminary phase, we have reduced the problem 
to sorting a smaller sequence $\sequencep$ (still of $\Oh{n}$
size) using various sequences built to be used with the
basic techniques in \S~\ref{sec:unstable}.  
Let $\distel$ be the number of distinct elements in $\sequencep$ 
(computed during the preliminary phase). 
The objective of this phase
is to sort each subsequence of size $\Thetah{\distel\polylog{n}}$ of
the main
sequence $\sequencep$. For any such subsequence $\sequencep_l$, we first find a
set of pivots and then sort $\sequencep_l$ with a distributive
approach. The guide zone is
sorted and is used to retrieve in sorted order lists of equal elements produced
by the distribution. The distribution zone provides sets of pivots elements
that are progressively moved into one of the spare zones. The head zone
furnishes placeholder elements for the distributive processes.
The distributive process depends on the hypothesis that each
$\distel$ contiguous elements of $\sequencep_l$ are sorted. The algorithm for
sorting $\Thetah{\distel}$ contiguous elements stably, in $\Oh{\distel}$ 
time and $\Oh{1}$ space (see \S~\ref{subsubsec:dist})
employs the pseudo pointers technique and the guide, jump, pointer and spare
zones are crucial in this process.

\bparagraph{Final phase (\S~\ref{subsec:final}).} 
After the aggregating phase the main sequence $\sequencep$ has all its
subsequences of $\Thetah{\distel\polylog{n}}$ elements in sorted order. With an
iterative merging process, we obtain from $\sequencep$ two new sequences: a
small sequence containing $\Oh{\distel\log^2 n}$ sorted elements; a large
sequence still containing $\Oh{n}$ elements but with an important
property: the
length of any subsequence of equal elements is multiple of a suitable number
$\Thetah{\log^2 n}$. By exploiting its property, the large sequence is sorted
using the encoded memory and merged with the small one. Finally, we
take care of all the zones built in the preliminary phase. Since they
have sizes either $\Oh{n/\log n}$ or $\Oh{n^\subrangeexp}$, they can be
easily sorted within our target bounds.

\subsection{Preliminary Phase}\label{subsec:prelim}
The preliminary phase has two main steps described in
Sections~\ref{subsubsec:enc}
and \ref{subsubsec:zones}.

\lsubsubsection{Encoded memory}\label{subsubsec:enc}
We start by collecting some pairs of distinct elements.
We go from $\sequence$ to a sequence
$Z'Y'XY''Z''$ with the same two properties we saw in
\S~\ref{subsubsec:bitstealing}. We use the linear time selection and
partitioning in \cite{KP92,KP94} which are also stable. 
Let us maintain the same notations used in \S~\ref{subsubsec:bitstealing}. 
We end up with $\memsize$ pairs of
distinct elements
$(M'[1],M''[1]),\ldots,(M'[\memsize],M''[\memsize])$ 
to encode information by bit stealing.
We use the encoded memory based on these pairs as if it were actual
memory (that is, we will allocate arrays, we will index and modify entries of
these arrays, etc).
However, in order not to lose track of the costs, the
names of encoded structures will always
be written in the following way: $\enca{I}$, $\enca{U}$, etc.

We allocate two
arrays $\enca{I}_{bg}$
and $\enca{I}_{en}$, each one with $\subrange=n^\subrangeexp$ entries of $1$ bit
each. The entries of both arrays are set to $0$.
$\enca{I}_{bg}$ and $\enca{I}_{en}$ will be used each time the procedure in
\S~\ref{subsubsec:dist} of the aggregating phase is invoked.

\lsubsubsection{Jump, guide, distribution, pointer,  
head  and spare zones}\label{subsubsec:zones}
The second main step of this phase has six sub-steps.

\inlinestep{First}. Let us suppose $\card{X}>n/\log n$ (otherwise
we sort it with the mergesort in \cite{SS87}). 
Using the selection and partitioning in \cite{KP92,KP94}, we
go from $X$ to $JX'$ such that $J$ is composed by the element $j^*$ 
with rank $3\subrange+1=3n^\subrangeexp+1$ (in $X$) plus all the elements
(in $X$) $\le j^*$. Then, we move the rightmost
element
equal to $j^*$ in the last 
position of $J$ (easily done in $\Oh{\card{J}}$ and stably with a sequence
exchange). 

\inlinestep{Second}. Let us suppose $\card{X'}>n/\log n$. 
With this step and the next one we separate the elements which appear more
than $7$ times. Let us allocate in
our encoded memory of $\memsize=\Oh{n/\log n}$ bits an encoded array $\enca{I}$
with $\subrange$ ($=n^\subrangeexp$) entries of
$4$ bits each. All the entries are initially set to $0$. Then, we start
scanning $X'$ from left to right. For any element $u\in X'$ accessed during the
scan, if $\enca{I}[u]\le 7$, we increment $\enca{I}[u]$ by one. 

\inlinestep{Third}. We scan $X'$ again. Let $u\in X'$
be the $i$th element accessed. If
$\enca{I}[u]<7$, we decrement $\enca{I}[u]$ by $1$ and exchange
$X'[i]$ ($=u$) with $J[i]$. At the end of the scan we have that
$J=WJ''$, where $W$ contains the elements of $X'$ occurring less than $7$ times
in $X'$. Then, we have to gather the elements previously in $J$ and now
scattered in $X'$. We accomplish this with the partitioning in
\cite{KP92}, using $J[\card{J}]$ to discern between elements previously in $J$
and the ones belonging to $X'$ (we know that $J[\card{J}]$ is equal to
$j^*$ and, for any $j\in J$ and any $x'\in X'$, $j\le J[\card{J}]<x'$).
After that we have $WJ''J'X''$ where the elements in $J'$ are the ones
previously in $J$ and exchanged during the scan of $X'$. 
We exchange $W$ with $J'$ ending up with $JWX''$.

\inlinestep{Fourth}. We know that each element of $X''$
occurs at least $7$ times in it. We also know that the entries of $\enca{I}$
encode either $0$
or $7$. 
We scan $X''$ from left to right. Let $u\in
X''$ be the $i$th element accessed. If
$\enca{I}[u]=7$, we decrement $\enca{I}[u]$ by one and we exchange
$X''[i]$ ($=u$) with $J[i]$. After the scan we have that
$J=GJ'''$, where, for any $j$, $G[j]$ was the leftmost occurrence of its kind in
$X''$ (before the scan). Then, we sort $G$ with the mergesort
in \cite{SS87} ($\card{G}=\Oh{\subrange}=\Oh{n^\subrangeexp}$ and
$\subrangeexp<1$). After that, similarly to the
third step, we gather the elements previously in $J$ and now scattered in
$X''$ because of the scan. We end up with the sequence $JWGX'''$.
We repeat the same process (only testing for $\enca{I}[u]$ equal to $6$
instead of $7$) to gather the leftmost occurrence of each
distinct element in $X'''$ into a zone $D$, ending up with the sequence
$JWGDX''''$.

\inlinestep{Fifth}. Each element of $X''''$ occurs at
least $5$ times in it and the entries of $\enca{I}$ encode either $0$
or $5$.
We scan $X''''$, let $u\in
X''''$ be the $i$th element accessed. If
$\enca{I}[u]=5$, we decrement $\enca{I}[u]$ by $1$ and exchange
$X''''[i]$ ($=u$) with $J[i]$. After the scan we have that
$J=PJ'''$, where, for any $j$, $P[j]$ was the leftmost occurrence of its kind in
$X''''$ (before the scan). Unlike the fourth step, we do not sort $P$.
We repeat the process finding $T_1$, $T_2$, $T_3$ and $H$ containing
the second, third, fourth and fifth leftmost occurrence of each distinct element
in $X''''$, respectively. After any of these processes, we gather back the
elements previously in $J$ scattered in $X''''$ (same
technique used in the third and fourth steps). 
We end up with the sequence
$JWGDPT_1T_2T_3H\sequencep$.

\inlinestep{Sixth}. Let us divide $J$ into $J_1J_2J_3V$, where
$\card{J_1}=\card{J_2}=\card{J_3}=\subrange$ and
$\card{V}=\card{J}-3\subrange$. We scan $G$, let $u\in G$ be the
$i$th element accessed, we exchange $T_1[i]$,
$T_2[i]$ and $T_3[i]$ with $J_1[u]$, $J_2[u]$ and $J_3[u]$, respectively.

\lsubsubsection{Summing up}
We will refer to $J_1$, $J_2$ and $J_3$ as \emph{jump zones}. Zone $G$, $D$, $P$
and $H$ will be referred to as \emph{guide}, \emph{distribution}, 
\emph{pointer} and \emph{head zones}, respectively. Finally, $T_1$, $T_2$ and
$T_3$ will be called \emph{spare zones}.
With the preliminary phase we have passed from the initial sequence
$\sequence$ to $M'J_1J_2J_3VWGDPT_1T_2T_3H\sequencep M''$. 
We allocated in
the encoded memory two arrays $\enca{I}_{bg}$ and $\enca{I}_{en}$.
The encoded memory, $\enca{I}_{bg}$ and $\enca{I}_{en}$, and the
jump, guide, distribution, pointer, head and spare zones
will be used in the next two phases to sort the $\sequencep$.
Zones $V$ and $W$ are a byproduct of the 
phase and will not have an active role in the sorting of $\sequencep$. 
The number of
distinct elements in sequence $\sequencep$ is less than or equal to the sizes of
guide, distribution, pointer and head zones. For the rest of the
paper we will denote $\card{G}$ ($=\card{D}=\card{P}=\card{H}$) with $\distel$.

\begin{lemma}\label{lem:prelim}
The preliminary phase requires $O(n)$ time, uses $O(1)$
auxiliary words and is stable.
\end{lemma}

\subsection{Aggregating Phase}\label{subsec:aggr}
Let us divide $\sequencep$ into $\numaggr$ subsequences
$\sequencep_1\sequencep_2\ldots\sequencep_\numaggr$ with
$\card{\sequencep_i}=\distel\log^\polyexp n$, for a suitable constant
$\polyexp\ge 4$. 
Let $\distfan=\log^\polydist n$, for a suitable constant $\polydist<1$. 
We will assume that $\distel\ge(2\distfan+1)\log\card{\sequencep_i}$. We leave
the particular case where $\distel< (2\distfan+1)\log\card{\sequencep_i}$ for
the full paper. 
For a generic $1\le l\le \numaggr$, let us assume that any $\sequencep_{l'}$
with $l'<l$ has been sorted and that $H$ is next to the
left end of $\sequencep_l$. To sort $\sequencep_l$ we have 
two main steps described in \S~\ref{subsubsec:pivots} and
\S~\ref{subsubsec:sort}. 
They rely on the algorithm described in \S~\ref{subsubsec:dist}.

\lsubsubsection{Sorting $\Oh{\distel}$ contiguous
elements}\label{subsubsec:dist}
Let us show how to exploit the two arrays $\enca{I}_{bg}$ and
$\enca{I}_{en}$, (in the encoded memory in the preliminary phase)
and the jump, guide and pointer zones to sort a sequence $A$, with
$\card{A}\le\distel$, stably in $\Oh{\card{A}}$ time and using $\Oh{1}$
auxiliary words. The process
has two steps.

\inlinestep{First}. We scan $A$, let $u\in A$ be the
$i$th element accessed. Let $p=P[i]$ and $h=J_1[u]$. 
If $\enca{I}_{bg}[u]=0$, we set both $\enca{I}_{bg}[u]$ and
$\enca{I}_{en}[p]$ to $1$. 
In any case, we exchange $J_1[u]$ ($=h$) with $J_2[p]$ and $A[i]$ 
($=u$) with $J_3[p]$. Then, we exchange $P[i]$ ($=p$) with $J_1[u]$
(which is not $h$ anymore).  

\inlinestep{Second}. Let $j=\card{A}$. We
scan $G$, let $g\in G$ be the
$i$th element accessed. If $\enca{I}_{bg}[g]=0$, we do
nothing. Otherwise, let $p$ be $J_1[g]$, we set $\enca{I}_{bg}[g]=0$ and
execute the following three steps. 
$(i)$ We exchange $J_3[p]$ with $A[j]$, then $J_1[g]$ with $P[j]$
and finally $J_1[g]$ with $J_2[p]$. $(ii)$ We decrease $j$ by $1$.
$(iii)$ If $\enca{I}_{en}[p]=1$, we set $\enca{I}_{en}[p]=0$ and the
sub-process ends, otherwise, let $p$ be $J_1[g+1]$, and we go to $(i)$.

Let us remark that the $\Oh{\card{A}}$ entries of 
$\enca{I}_{bg}$ and $\enca{I}_{en}$  that are changed from $0$ to $1$ in the
first step, are set back to $0$ in the second one. 
$\enca{I}_{bg}$ and
$\enca{I}_{en}$ have been initialized in the preliminary phase.
We could not afford to re-initialize them every time we
invoke the process (they have $r=n^\subrangeexp$
entries and $\card{A}$ may be $\oh{n^\subrangeexp}$).

\begin{lemma}\label{lem:dist}
Using the encoded arrays $\enca{I}_{bg}$, $\enca{I}_{en}$ and the jump, guide
and pointer zones, a sequence $A$ with $\card{A}\le\distel$ can be sorted
stably, in $O(\card{A})$ time and using $O(1)$ auxiliary words.
\end{lemma}

\lsubsubsection{Finding pivots for $\sequencep_l$}\label{subsubsec:pivots}
We find a set of pivots $\set{e_1,e_2\ldots,e_{p-1},e_p}$
with the following properties: 
$(i)$ $\card{\set{x\in\sequencep_l\,|\,x<e_1}}\le\distel$;
$(ii)$ $\card{\set{x\in\sequencep_l\,|\,e_p<x}}\le\distel$;
$(iii)$ $\card{\set{x\in\sequencep_l\,|\,e_i<x<e_{i+1}}}\le\distel$, for any
$1\le i<p$; 
$(iv)$ $p=\Thetah{\log^\polyexp n}$. 
In the end the pivots reside in the first $p$ positions of $D$.
We have four steps.

\inlinestep{First}. We allocate in the encoded memory an array $\enca{P}$ with
$\subrange$ entries of $\log \card{\sequencep_l}$ bits, but we \emph{do not
initialize each entry of} $\enca{P}$. We
initialize the only $\distel$ of them we will need: for any
$i=1\ldots\card{G}$, we set $\enca{P}[G[i]]$ to $0$. 

\inlinestep{Second}. We scan $\sequencep_l$ from left to
right. Let $u\in\sequencep_l$ be the $i$th element accessed, we increment
$\enca{P}[u]$ by $1$. 

\inlinestep{Third}. We sort the distribution zone $D$ using the algorithm
described in
\S~\ref{subsubsec:dist}.

\inlinestep{Fourth}. Let $i=1$, $j=0$ and $p=0$. We repeat the following process
until
$i>\card{G}$. $(i)$ Let $u=G[i]$, we set $j=j+\enca{P}[u]$. $(ii)$ If
$j<\distel$
we increase $i$ by $1$ and go to $(i)$. If $j\ge\distel$, we increase $p$
by $1$, exchange $D[i]$ with $D[p]$, increase $i$ by $1$, set $j$ to $0$ and
go to $(i)$.

\lsubsubsection{Sorting $\sequencep_l$}\label{subsubsec:sort}
Let $p$ be the number of pivots for $\sequencep_l$ selected in
\S~\ref{subsubsec:pivots} and now residing in the first $p$
positions of $D$. Let $\unitsize$ be $\log\card{\sequencep_l}$.
Let us assume that $H$ is next to the left end of $\sequencep_i$.
We have six steps.

\inlinestep{First}. Let $i=0$ and $j=0$. The following two steps are repeated
until $i>p$:
$(i)$ we increase $i$ by $p/\distfan$ and $j$ by $1$; $(ii)$ we exchange $D[i]$
with $T_1[j]$. We will denote with $p'$ the number of selected
pivots, now temporarily residing in the first $p'$ positions of $T_1$.

\inlinestep{Second}. Let us divide $\sequencep_l$ into
$q=\card{\sequencep_l}/\distel$ blocks $B_1B_2\ldots B_q$ of $\distel$
elements each. We sort each $B_i$ using the algorithm in
\S~\ref{subsubsec:dist}.

\inlinestep{Third}. With a sequence exchange we bring $H$ next to the right end
of
$\sequencep_l$. Let us divide $H$ into 
$H_1\hat{H}_1\ldots H_{p'}\hat{H}_{p'}H_{p'+1}H'$,
where $\card{H_{p'+1}}=\card{H_i}=\vert\hat{H}_i\vert=\unitsize$. Let
$f=\card{\sequencep_l}/\unitsize$. We
allocate the following arrays: $(i)$
$\enca{U}_{suc}$ and
$\enca{U}_{pre}$ both with $f+2p'+1$ entries of $\Thetah{\unitsize}$
bits; $(ii)$ $\enca{H}$ and $\hat{\enca{H}}$ with $p'+1$ and $p'$ entries of
$\Thetah{\log \unitsize}$ bits; $(iii)$ $\enca{L}$ and $\hat{\enca{L}}$
with $p'+1$ and $p'$ entries of $\Thetah{\unitsize}$ bits; $(iv)$
$\enca{N}$ and $\hat{\enca{N}}$ with $p'+1$ and $p'$
entries of $\Thetah{\unitsize}$ bits. 
Each entry of any array is initialized to $0$. 

\inlinestep{Fourth}. In this step we want to transform $\sequencep_l$ and $H$ in
the
following ways. We pass from $\sequencep_l$ to 
$U_1U_2\ldots U_{f'-1}U_{f'}H''$, where the $U_i$'s are called \emph{units}, 
for which the following holds. 
\begin{itemize}

\item[$(i)$] $f'\ge f-(2p'+1)$ and
$\card{U_i}=\unitsize$, for any $1\le i \le f'$. 

\item[$(ii)$] For any $U_i$, $1\le i
\le f'$, one of the 
following holds: $(a)$ there exists a $1\le j\le p'$ such that
$x=T_1[j]$, for any $x\in U_i$; $(b)$ $x<T_1[1]$, for
any $x\in U_i$; $(c)$ $T_1[p']<x$, for any $x\in U_i$; $(d)$ there exists a 
$1\le j'\le p'-1$ such that $T_1[j']<x<T_1[j'+1]$, for any $x\in U_i$.

\item[$(iii)$] Let us call a \emph{set of
related units} a maximal set of units
$\seta{U}=\set{U_{i_1},U_{i_2},\ldots,U_{i_z}}$ for which one of the 
following conditions holds: $(a)$ there exists a $1\le j\le p'$ such that
$x=T_1[j]$, for any $x\in U_i$ and for any $U_i\in\seta{U}$; $(b)$ $x<T_1[1]$,
for any $x\in U_i$ and for any $U_i\in\seta{U}$; $(c)$ $T_1[p']<x$, for any
$x\in U_i$
and for any $U_i\in\seta{U}$; $(d)$ there exists a $1\le j'\le p'-1$
such that $T_1[j']<x<T_1[j'+1]$, for any $x\in U_i$ and for any
$U_i\in\seta{U}$.
For any set of related units $\seta{U}=\set{U_{i_1},U_{i_2},\ldots,U_{i_z}}$
we have that $\enca{U}_{suc}[i_y]=i_{y+1}$ and
$\enca{U}_{pre}[i_{y+1}]=i_{y}$, for
any $1\le y\le z-1$.
\end{itemize}
Concerning $H''$ and $H=H_1\hat{H}_1\ldots H_{p'}\hat{H}_{p'}H_{p'+1}H'$. Before
this step all the elements in $H$ were the original ones gathered in
\S~\ref{subsubsec:zones}. After
the fourth step, the
following will hold. 
\begin{itemize}
\item[$(iv)$] The elements in $H'$ and $H''$ plus the elements in
$H_i[\enca{H}[i]+1\ldots \unitsize]$ and 
in $\hat{H}_{i'}[\hat{\enca{H}}[i']+1\ldots \unitsize]$, for any $1\le i\le
p'+1$ and
$1\le i'\le p'$, form the original set of elements that were in $H$ before
the fourth step. 

\item[$(v)$] We have that: $(a)$ $x<T_1[1]$, for any $x\in
H_1[1\ldots\enca{H}[1]]$; $(b)$ $x>T_1[p']$, for any 
$x\in H_{p'+1}[1\ldots\enca{H}[p'+1]]$; $(c)$ $T_1[i-1]<x<T_1[i]$, for any 
$x\in H_{i}[1\ldots\enca{H}[i]]$ and any $2\le i\le p'$; 
$(d)$ $x=T_1[i]$, for any 
$x\in \hat{H}_{i}[1\ldots\hat{\enca{H}}[i]]$ and any $1\le i\le p'$.

\item[$(vi)$] $(a)$ Let $j=\enca{L}[1]$ ($j=\enca{L}[p']$), $U_j$ is the
rightmost unit such that $x<T_1[1]$ ($x>T_1[p']$) for any $x\in U_j$. $(b)$
For any $2\le i\le p'$, let $j=\enca{L}[i]$, $U_j$ is the rightmost unit
such that $T_1[i-1]<x<T_1[i]$ for any $x\in U_j$. 
$(c)$ For any $2\le i\le p'$, let $j=\hat{\enca{L}}[i]$, $U_j$ is the rightmost
unit such that $x=T_1[i]$ for any $x\in U_j$. 

\item[$(vii)$] $(a)$ $\enca{N}[1]$ ($\enca{N}[p']$) is the number of 
$x\in\sequencep_l$ 
such that $x<T_1[1]$ ($x>T_1[p']$). $(b)$
For any $2\le i\le p'$, $\enca{N}[i]$ is the number of $x\in\sequencep_l$ 
such that $T_1[i-1]<x<T_1[i]$. 
$(c)$ For any $2\le i\le p'$, $\hat{\enca{N}}[i]$ is the number of
$x\in\sequencep_l$ such that $x=T_1[i]$. 

\end{itemize}
Let $h=\enca{H}[1]$, let $i=1$ and let $j=1$. We start the fourth step
by scanning $B_1$. If
$B_1[i]<T_1[1]$, we increase $h$, $\enca{H}[1]$ and $\enca{N}[1]$ by $1$,
exchange $B[i]$ with $H_1[h]$ and increase $i$ by $1$. This sub-process goes on
until one of the following two events happens: $(a)$ $h$ and $\enca{H}[1]$ are
equal to $\unitsize+1$; $(b)$ $B_1[i]\ge T_1[1]$. If event $(a)$ happens, we
exchange the $\unitsize$ elements currently in $H_1$ with
$\sequencep_l[(j-1)\unitsize+1\ldots j\unitsize]$. Then,
we set $\enca{U}_{pre}[j]$ to $\enca{L}[1]$, $\enca{U}_{suc}[\enca{L}[1]]$ to
$j$ and $\enca{L}[1]$ to $j$. After that, we set $h$ and $\enca{H}[1]$ to
$0$ and we increment $j$ by $1$. Finally, we go back to the scanning of
$B_1$. Otherwise, if event $(b)$ happens, we set $h$ to $\hat{\enca{H}}[1]$ and
we continue the scanning of $B_1$ but with the following sub-process: if
$B_1[i]=T_1[1]$, we increase $h$, $\hat{\enca{H}}[1]$ and $\hat{\enca{N}}[1]$ by
$1$, exchange $B[i]$ with $\hat{H}_1[h]$ and increase $i$ by $1$. In its turn,
this sub-process goes on until one of the following two events happens: $(a')$
$h$ and $\hat{\enca{H}}[1]$ are equal to $\unitsize+1$; $(b')$ $B_1[i]>T_1[1]$.
Similarly to what we did for event $(a)$, if event $(a')$ happens, we exchange
the $\unitsize$ elements currently in $\hat{H}_1$ with
$\sequencep_l[(j-1)\unitsize+1\ldots j\unitsize]$. Then, we set
$\enca{U}_{pre}[j]$ to $\hat{\enca{L}}[1]$, $\enca{U}_{suc}[\hat{\enca{L}}[1]]$
to $j$ and $\hat{\enca{L}}[1]$ to $j$. After that, we set $h$ and
$\hat{\enca{H}}[1]$ to $0$ and we increment $j$ by $1$. Finally, we go back to
the scanning of $B_1$. Otherwise, if event $(b')$ happens, we
set $h$ to $\enca{H}[2]$ and we
continue the scanning of $B_1$ but with the following sub-process: if
$B_1[i]<T_1[2]$, we increase $h$, $\enca{H}[2]$ and $\enca{N}[2]$ by $1$,
exchange $B[i]$ with $H_2[h]$ and increase $i$ by $1$. We continue in this
fashion, possibly passing to $\hat{\enca{H}}[2]$, $\enca{H}[3]$, etc, until
$B_1$ is exhausted. Then, the whole process is applied to $B_2$
from the beginning. When $B_2$ is exhausted we pass to $B_3$, $B_4$ and so forth
until each block is exhausted. 

\inlinestep{Fifth}. We start by exchanging $H''$ and $H_1\hat{H}_1\ldots
H_{p'}\hat{H}_{p'}H_{p'+1}$. 
Let $h=1$ and $h'=\unitsize-\enca{H}[1]$. We exchange the elements in
$H_1[\enca{H}[1]+1\ldots \unitsize]$ with the ones in $T_2[h\ldots h']$.
After that, we set $h=h'$, increment $h'$ by $\unitsize-\hat{\enca{H}}[1]$ and
exchange
the elements in $\hat{H}_1[\hat{\enca{H}}[1]+1\ldots \unitsize]$ with the ones
in
$T_2[h\ldots h']$. We proceed in this fashion until $H_{p'+1}$ is done. Then we
exchange the elements in $H''H'$ with the rightmost $\card{H''H'}$ ones in
$T_2$. After that we execute the following process to ``link'' the $H_i$'s and
$\hat{H}_i$'s to their respective sets of related units. We start by setting
$\enca{U}_{pre}[f'+1]$ to $\enca{L}[1]$ and $\enca{U}_{suc}[\enca{L}[1]]$ to
$f'+1$. Then we set $\enca{U}_{pre}[f'+2]$ to $\hat{\enca{L}}[1]$ and
$\enca{U}_{suc}[\hat{\enca{L}}[1]]$ to $f'+2$. We proceed in this fashion until
we set $\enca{U}_{pre}[f'+2p'+1]$ to $\enca{L}[p'+1]$ and
$\enca{U}_{suc}[\enca{L}[p'+1]]$ to $f'+2p'+1$.
After that we execute the following process to bring each set
of related units into a contiguous zone.
Let $i=\flo{\enca{N}[1]/\unitsize}+1$, we exchange $H_{1}$ with
$U_i$; we swap the values in $\enca{U}_{pre}[i]$ and $\enca{U}_{pre}[f'+1]$,
then the values in $\enca{U}_{suc}[i]$ and $\enca{U}_{suc}[f'+1]$; we set
$\enca{U}_{pre}[\enca{U}_{suc}[f'+1]]=f'+1$ and
$\enca{U}_{suc}[\enca{U}_{pre}[f'+1]]=f'+1$; finally, we set 
$\enca{U}_{suc}[\enca{U}_{pre}[i]]=i$. Then, let $j=\enca{U}_{pre}[i]$, we
decrement $i$ by $1$, we exchange $U_j$ with $U_i$; we swap the values in
$\enca{U}_{pre}[i]$ and $\enca{U}_{pre}[j]$,
then the values in $\enca{U}_{suc}[i]$ and $\enca{U}_{suc}[j]$; we set
$\enca{U}_{pre}[\enca{U}_{suc}[j]]=j$ and
$\enca{U}_{suc}[\enca{U}_{pre}[j]]=j$; finally, we set
$\enca{U}_{pre}[\enca{U}_{suc}[i]]=i$ and
$\enca{U}_{suc}[\enca{U}_{pre}[i]]=i$. We proceed in this fashion until the
entire set of related units of $H_1$ resides in 
$\sequencep_l[1\ldots\enca{N}[1]-\enca{H}[1]]$ ($H_1$ now resides in
$\sequencep_l[\enca{N}[1]-\enca{H}[1]+1\ldots
\enca{N}[1]-\enca{H}[1]+\unitsize]$). 
After that, we
apply the same process to $\hat{H}_1$, $H_2$, $\hat{H}_2$ and so forth until
every set of related units has been compacted into a contiguous zone. We end up
with the sequence
$\overline{U}_1H_1\hat{\overline{U}}_1\hat{H}_1\ldots
\overline{U}_{p'}H_{p'}\hat{\overline{U}}_{p'}\hat{H}_{p'}\overline{U}_{p'+1}H_{
p'+1}H''H'$, where each $\overline{U}_i$ contains the set of related units of
$H_i$ and each $\hat{\overline{U}}_i$ the set of related units of $\hat{H}_i$.
Finally, we proceed to separate
the elements of $\sequencep_l$ from the ones residing in 
$H_i[\enca{H}[i]+1\ldots \unitsize]$ and $\hat{H}_{i'}[\enca{H}[i']+1\ldots
\unitsize]$, for
any $1\le i\le p'+1$ and any $1\le i'\le p'$. Since the ``intruders''
were previously residing in $T_2$, by the first and sixth steps in
\S~\ref{subsubsec:zones} we know that any of them is less than any of the
elements of $\sequencep_l$. Therefore we can separate them with the
stable partitioning algorithm in \cite{KP92} and end up with the sequence
$R_1\hat{R}_1\ldots R_{p'}\hat{R}_{p'}R_{p'+1}H'''H''H'$. Finally, we exchange
the elements in the sequence $H=H'''H''H'$ with the ones in $T_2$, getting back
in $H$ its original (distinct) elements.

\inlinestep{Sixth}. After the fifth step we are left with the sequence 
$R_1\hat{R}_1\ldots R_{p'}\hat{R}_{p'}R_{p'+1}H$ for which the following holds:
$(i)$ $\enca{N}[i]=\card{R_i}$ and $\hat{\enca{N}}[i']=\card{R_{i'}}$, for any
$1\le i\le p'+1$ and any $1\le i'\le p'$;
$(ii)$ for any $x\in R_1$, $x<T_1[1]$; $(iii)$ for any $x\in R_{p'+1}$,
$x<T_1[p']$; $(iv)$ for any $x\in\hat{R}_i$, $x=T_1[i]$ ($1\le i\le p'$);
$(v)$ for any $x\in R_i$, with $1\le i<p'$, $T_1[i]<x<T_1[i+1]$. 
We begin by moving $H$ before $R_1$ with a sequence exchange.
Since we do not need the $p'$ pivots anymore, we put them back in their
original positions in $D[1\ldots p]$ executing once again the process in the
first step. Then, we allocate in the encoded memory an array $\enca{R}$
with $p'+1$ entries of $\Thetah{\unitsize}$ bits. 
Let $i=1$ and $\enca{R}[1]=1$: for
$j=2,\ldots,p'+1$ we increment $i$ by $\enca{N}[j-1]+\hat{\enca{N}}[j-1]$ and
set $\enca{R}[j]=i$. After that, we execute a series of $p'+1$ recursive
invocations of the procedure here in
\S~\ref{subsubsec:sort} (not the whole sorting algorithm).
We start by sorting $R_1=\sequencep_l[\enca{R}[1]\ldots\enca{R}[2]-1]$
recursively with the same procedure in this section: we use
$\sequencep_l[\enca{R}[1]\ldots\enca{R}[2]-1]$ in place of $\sequencep_l$,
$D[1\ldots p/t-1]$ in place of $D$ and so forth. After $R_1$ is sorted, we
swap $H$ and $R_1$ with a sequence exchange and proceed to sort
$R_2$: we use $\sequencep_l[\enca{R}[2]\ldots\enca{R}[3]-1]$ in place of
$\sequencep_l$, $D[p/t+1\ldots 2p/t-1]$ in place of $D$ and so forth. We
proceed in this fashion until $R_{p'+1}$ is sorted and $H$ is located
after it again. We do not need anything particularly complex to handle the
recursion with $\Oh{1}$ words since there can be only $O(1)$ nested
invocations.

\begin{lemma}\label{lem:aggr}
The aggregation phase requires $O(n)$ time, uses $O(1)$
auxiliary words and is stable.
\end{lemma}

\subsection{Final Phase}\label{subsec:final}
The final phase has two main steps described in
\S~\ref{subsubsec:sortsequencep} and \S~\ref{subsubsec:sortzones}.

\lsubsubsection{Sorting the main sequence
$\sequencep$}\label{subsubsec:sortsequencep}
After the aggregating phase we are left with
$\sequencep=\sequencep_1\sequencep_2\ldots\sequencep_\numaggr$ where,
for any $1\le i\le\numaggr$,
$\sequencep_i$ is sorted and $\card{\sequencep_i}=\distel\log^\polyexp n$. 
Let $\finalunitsize=\log^2 n$. We have three steps.

\inlinestep{First}. We allocate in the encoded memory an array
$\enca{\sequence}'_1$ with
$\card{\sequencep_1}$ entries of one bit initially
set to $1$. Then, we scan $\sequencep_1$. During the scan,
as soon as we encounter a subsequence $\sequencep_1[i\ldots
i+\finalunitsize-1]$ of equal elements (that is a subsequence of
$\finalunitsize$ consecutive equal elements) we set $\enca{\sequence}'_1[i]$,
$\enca{\sequence}'_1[i+1]$\ldots$\enca{\sequence}'_1[i+\finalunitsize-2]$ and
$\enca{\sequence}'_1[i+\finalunitsize-1]$ to $0$. After that we use the
partitioning algorithm in \cite{KP92} to separate the
elements of $\sequencep_1$ with the corresponding entries in $\enca{\sequence}'$
set to $1$ from the ones with their entries set to $0$ (during the
execution of the partitioning algorithm in \cite{KP92}, each time two elements
are exchanged, the values of their entries in $\enca{\sequence}'$ are 
exchanged too). After the partitioning we have
that $\sequencep_1=\sequences_1O_1$ and the following conditions hold: $(i)$ 
$\sequences_1$ and $O_1$ are still sorted ($\sequencep_1$ was sorted and the
partitioning is stable); $(ii)$ the length of any maximal subsequence of
consecutive equal elements of $\sequences_1$ is a multiple of $\finalunitsize$;
$(iii)$ $\card{O_1}\le \distel\finalunitsize$.
Then, we merge $O_1$ and $\sequencep_2$ using the merging
algorithm in \cite{SS87}, obtaining a sorted sequence $\sequencet_2$
with $\card{\sequencet_2}<2\distel\log^\polyexp n$. After that, we apply to
$\sequencet_2$ the same process we applied to $\sequencep_1$ ending up with
$\sequencet_2=\sequences_2O_2$ where conditions $(i)$, $(ii)$ and $(iii)$ hold
for $\sequences_2$ and $O_2$ too. We proceed in this fashion until each
$\sequencep_\numaggr$ is done.

\inlinestep{Second}. We have that $\sequencep=\sequences O$.
The following conditions hold: $(i)$ $O$ is sorted and $\card{O}\le
\distel\finalunitsize$; $(ii)$ the length of any maximal subsequence of
consecutive equal elements of $\sequences_1$ is a multiple of $\finalunitsize$
(and so $\card{\sequences}$ is a multiple of $\finalunitsize$ too).
Let $\finalunitnum=\card{\sequences}/\finalunitsize$ and let us
divide $\sequences$ into $\finalunitnum$ subsequences $F_1F_2\ldots
F_{\finalunitnum-1}F_\finalunitnum$ with $\card{F_i}=\finalunitsize$.
We allocate in the encoded memory two arrays $\enca{\sequence}''_{pre}$ and
$\enca{\sequence}''_{suc}$, each one with $\finalunitnum$
entries of $\Thetah{\log n}$ bits. We also allocate an array $\enca{C}$
with $\subrange$ ($=n^\subrangeexp$) entries of $\Thetah{\log n}$ bits, each
one initialized to $0$.
Then, for each $F_i$ from the rightmost to the leftmost one, we do the
following: Let $v=F_i[1]$. If $\enca{C}[v]=0$, we set
$\enca{\sequence}''_{suc}[i]=0$ and $\enca{C}[v]=i$. Otherwise, if
$\enca{C}[v]\not=0$, we set $\enca{\sequence}''_{suc}[i]=\enca{C}[v]$,
$\enca{\sequence}''_{pre}[\enca{C}[v]]=i$ and $\enca{C}[v]=i$.

\inlinestep{Third}. We scan $\enca{C}$ and find the leftmost entry not equal to
$0$, let it
be $i$. Let $j=\enca{C}[i]$, we exchange $F_1$ with $F_j$ and do the following:
$(i)$ we swap the values in
$\enca{\sequence}''_{pre}[1]$ and $\enca{\sequence}''_{pre}[j]$,
then the values in $\enca{\sequence}''_{suc}[1]$ and
$\enca{\sequence}''_{suc}[j]$; 
$(ii)$ we set
$\enca{\sequence}''_{pre}[\enca{\sequence}''_{suc}[j]]=j$ and
$\enca{\sequence}''_{suc}[\enca{\sequence}''_{pre}[j]]=j$; 
$(iii)$ finally, we set
$\enca{\sequence}''_{pre}[\enca{\sequence}''_{suc}[1]]=1$ and
$\enca{\sequence}''_{suc}[\enca{\sequence}''_{pre}[1]]=1$.
Then, let $j'=\enca{\sequence}''_{suc}[1]$. We exchange $F_2$ and $F_{j'}$ and
then we make similar adjustments to their entries in $\enca{\sequence}''_{pre}$
and $\enca{\sequence}''_{suc}$. We proceed in this fashion until we exhaust the
linked list associated with the $i$th entry of $\enca{C}$. After that we
continue to scan $\enca{C}$, find the leftmost non-zero entry $i'>i$ and
process its associated list in the same way. 
At the end of the process, the $F_i$'s have been permuted in sorted 
stable order. Now both $\sequences$ and $O$ are sorted. We merge
them with the merging algorithm in \cite{SS87}
and $\sequencep$ is finally sorted. 

\lsubsubsection{Taking care of the encoded memory and the zones}%
\label{subsubsec:sortzones}
In the last main step of the final phase we sort all those
zones that have been built in the preliminary phase. 
%
We scan $G$, let $u\in G$ be the $i$th element
accessed, we exchange $T_1[i]$,
$T_2[i]$ and $T_3[i]$ with $J_1[u]$, $J_2[u]$ and $J_3[u]$, respectively. 
%
We swap $\sequences$ and $H$ ($H$
has been moved after $\sequences$ at the end of the aggregating phase). $W$,
$G$, $D$, $P$, $T_1$, $T_2$, $T_3$ and $H$
have $\Oh{\distel}=\Oh{n^\subrangeexp}$ elements
and we can sort them using the mergesort in
\cite{SS87}. The obtained sequence can now be merged with $\sequences$ using the
merging algorithm in \cite{SS87}.
%
$J_1$, $J_2$, $J_3$ and $V$ have $\Oh{n^\subrangeexp}$
elements. We sort them using the mergesort in~\cite{SS87}. 
Finally, $M'$ and $M''$ have $\Thetah{n/\log n}$ elements. We sort them with
the mergesort in~\cite{SS87} and we are done. 

\begin{lemma}\label{lem:final}
The final phase requires $O(n)$ time, uses $O(1)$
auxiliary words and is stable.
\end{lemma}


\fi

\iffull
\section{Acknowledgements}
\else
\paragraph{Acknowledgements.}
\fi
Our sincere thanks to Michael Riley and Mehryar Mohri of 
Google, NY who, motivated by manipulating transitions of large finite state 
machines, asked if bucket sorting can be done in-place
in linear time. 

\bibliographystyle{plain}
\bibliography{biblio}

\end{document}